\documentclass{LMCS}

\def\doi{8 (1:03) 2012}
\lmcsheading%
{\doi}
{1--24}
{}
{}
{May~\phantom.16, 2011}
{Feb.~16, 2012}
{}

\usepackage{enumerate}
\usepackage{hyperref}
 \usepackage{macros}
 \usepackage{localmacros}
 \usepackage{color}
 \usepackage{amssymb}
 \usepackage{stmaryrd}
 \usepackage{bussproofs}

\theoremstyle{plain}
\theoremstyle{plain}\newtheorem{lemma}[thm]{Lemma}
\theoremstyle{definition}\newtheorem{definition}[thm]{Definition}
\theoremstyle{definition}\newtheorem{remark}[thm]{Remark}
\theoremstyle{plain}\newtheorem{theorem}[thm]{Theorem}

\DeclareRobustCommand*\cal{\relax\mathcal}

\begin{document}

\title[On completeness of reducibility candidates]{On completeness of reducibility candidates as a semantics of strong normalization}

\author[D.~Cousineau]{Denis Cousineau}	
\address{INRIA Saclay - \^Ile de France.}	
\email{denis@cousineau.eu}  
\thanks{This work was partly supported by INRIA-Microsoft Research joint lab, France.}	






\begin{abstract}
 
 This paper defines a sound and complete semantic criterion, based on reducibility candidates, for strong normalization of theories expressed in minimal deduction modulo {\em \`a la Curry}. The use of Curry-style proof-terms allows to build this criterion on the classic notion of pre-Heyting algebras and makes that criterion concern {\em all} theories expressed in minimal deduction modulo. Compared to using Church-style proof-terms, this method provides both a simpler definition of the criterion and a simpler proof of its completeness.
 
  \end{abstract}

\maketitle

\section*{Introduction}\label{S:one}


\medskip

\noindent In 1936, Tarski was the first to {\em formally} exhibit a link between model theory and proof theory \cite{Tarski,Tarskitrans}. Model theory is the study of {\em what is semantically true}, via the study of algebraic structures. Whereas proof theory is the study of {\em what is syntactically provable}, via the study of logical systems. Tarski showed that we can deduce properties in the syntactic world from properties in the semantic world. In particular, he proved that consistency of a theory  in first order logic is entailed by the existence of a model on some Boole algebra for that theory ({\em i.e.} a function, satisfying some properties, from the syntactic language to particular semantic objects).
Consistency is a syntactic property that ensures that a theory does not contain any contradiction. If two propositions are contradictory, then there is at least one of them that we cannot prove in a consistent theory. Tarski's result shows that the existence of such a Boole-valued model forms a {\em sound} (semantic) criterion for (syntactic) consistency of first order theories.	

\medskip

\noindent On the other hand, G\"odel had previously shown, in its completeness theorem \cite{Godel}, that we can also deduce semantic properties from syntactic ones. This theorem states that if a theory is consistent in first order logic, then one can build a Boole-valued model for that theory. Hence the existence of a Boole-valued model forms a sound {\em and complete} criterion for consistency of first order theories.
This link in both directions between semantic and syntactic worlds is a fundamental tool to study consistency of logic systems. For example, it allows to prove that the axiom of choice is independent from Zermelo-Fraenkel set theory (ZF). This means that adding the axiom of choice or its negation to the axioms of ZF does not change consistency. 
The idea of the proof is to first exhibit, by completeness and the hypothesis of syntactic consistency of ZF, a Boole-valued model of ZF in the semantic world. Then semantic tools can be used to transform it into a Boole-valued model of ZF with the axiom of choice, ZFC ({\em resp.} ZF with the negation of the axiom of choice, ZF$\neg$C). Finally, it entails consistency of ZFC ({\em resp.} ZF$\neg$C), by soundness. The ZF to ZFC part was proved by G\"odel \cite{GodelcontCH} while the ZF to ZF$\neg$C part was proved by Cohen \cite{CohenCH}.

\medskip

\noindent From a computer scientist point of view, consistency is not a sufficient property. Indeed, in order to define theories in which all (constructive) proofs can be machine-checked, a stronger property is needed, namely {\em cut elimination}, meaning that all proofs can be represented in a canonical {\em cut-free} way.

%

%

\medskip

\noindent Deduction modulo~\cite{DHK} is a generic way to integrate
computation rules into a deduction system, such as natural deduction
or sequent calculus.
In this paper, we shall only consider the case of natural deduction.
Deduction modulo can express theories with rewrite rules instead of axioms,
({\em e.g.} 
Peano arithmetic~\cite{GDowBWer05}, higher-order
logic~\cite{DHK} and Zermelo set
theory~\cite{GDowAMiq07}). Expressing axioms via rewrite rules allows to express the notion of
an axiomatic cut through a combination of regular cuts, leading to 
a uniform notion for cuts and cut elimination.
This 
gives a  
generic method to prove cut elimination for theories expressed in
natural deduction modulo, which consists in proving the strong normalization
property for the corresponding proof-terms, via the proofs-as-programs
paradigm (a.k.a.\ the Curry-Howard correspondence). In 
deduction modulo, if all proof-terms of a theory are strongly normalizing, 
then this theory satisfies the cut elimination property (and is moreover consistent in the case of {\em constructive} deduction modulo).

\medskip

\noindent In 1971, following the work of Tait \cite{Tait66}, Girard developed an apparently syntactic method for proving strong normalization, called {\em reducibility candidates} \cite{Girard71}. The main idea of this method is to associate to each proposition $A$, a set of strongly normalizing proofs and then prove that this set contains all the proofs of $A$. 
This method has been extended to several logical frameworks. In particular, Dowek and Werner defined reducibility candidates for deduction modulo \cite{DW}. Their extension provides a {\em sound} criterion for strong normalization of theories expressed in deduction modulo. Dowek also defined {\em pre-Heyting algebras} \cite{DOW-TVA}
and showed that reducibility candidates can be defined as a model valued on one of these algebras.
He provided this way a {\em semantic} sound criterion for strong normalization of theories expressed in deduction modulo
(when expressing proof-terms of deduction modulo with a system {\em \`a la Church}).

\medskip

\noindent 

\noindent Proof-terms \`a la Church differ from proof-terms \`a la Curry by
the amount of information, concerning the proving derivation, that is kept in the corresponding proof-term.
For example, in second-order logic, there are two type systems that capture
the computational contents of cut elimination: Church-style
system~$F$~\cite{Girard71} and Curry-style
system~$F$~\cite{Leivant83}.
The latter differs from the former in that it does not keep track of
introduction and elimination steps of second-order
quantification---and thus of the corresponding cuts.
Surprisingly, it can be shown by purely combinatorial means (as
opposed to semantic means) that the strong normalization property for
Church-style system~$F$ is equivalent to the strong normalization
property for its Curry-style variant.
Church-style or Curry-style proof-terms can indifferently be considered
to prove the cut elimination property for second-order logic, and in
practice, it is easier to prove the strong normalization property for
Curry-style system~$F$, which is based on notions of reducibility
candidates that are technically easier to manipulate than the
corresponding notions in Church's world.
A similar situation exists more generally between all the Pure Type
Systems~\cite{BG01} of the left-hand side of Barendregt's cube and the
corresponding \emph{Type Assignment Systems}~\cite{BLRU94} (their
Curry-style equivalents); or for second-order functional arithmetic
\cite{Kri93}, where both presentations coexist.
  In \cite{CouMiq}, with A. Miquel, we proved that strong normalization
of proof-terms \`a la Church and proof-terms \`a la Curry are equivalent for a large 
class of theories  ({\em non-confusing} theories, see definition \ref{confusion}) expressed in minimal deduction modulo
(deduction modulo where the language of
propositions uses implication as the only connective and (first-order) universal
quantification as the only quantifier) .

\medskip

\noindent The general purpose of my work is to define a sound and complete semantic criterion for strong normalization of theories expressed in minimal deduction modulo. In a previous paper \cite{pstt} (see also \cite{ Theseoim}), I extended the notion of pre-Heyting algebra and tuned the definition of reducibility candidates to obtain a sound and complete semantic criterion for strong normalization of {\em non-confusing} theories expressed in minimal deduction modulo {\em \`a la Church}. Considering proof-terms \`a la Church brought a lot of difficulties, in particular because the classic notion of pre-Heyting algebras cannot be used. This paper considers minimal deduction modulo {\em \`a la Curry}.
It provides a stronger result since it also applies to confusing theories. Moreover, using Curry-style deduction modulo allows to use the classic notion of pre-Heyting algebra, simplifying this way both the concept of those complete reducibility candidates and their proof of completeness.

\medskip 

\noindent This paper is organized as follows: in section 1, we first define minimal deduction modulo \`a la Church and \`a la Curry. Section 2  presents the concept of reducibility candidates and how to define them as a model valued on a pre-Heyting algebra (for minimal deduction modulo). In section 3, we explain how to tune the usual definition of reducibility candidates in order to obtain completeness while keeping soundness for strong normalization. And finally, section\ 4  formally defines this new notion of reducibility candidates as a model valued on a pre-Heyting algebra, and proves that it provides a sound and complete semantic criterion for strong normalization of theories expressed in minimal deduction modulo \`a la Curry (theorem \ref{compsound}).

\bigskip
\bigskip
\section{Minimal deduction modulo}


\noindent As in first-order logic, the language of a theory in deduction
modulo~\cite{HSchATro96} is obtained from a signature defining a set of
function symbols and a set of predicate symbols given with their
ranks, or arities.
For convenience, we shall only consider the case of mono-sorted
theories. We are convinced that results of the present paper trivially extend to many-sorted theories.
In minimal deduction modulo, propositions are then built-up from
predicates, with the only connective~$\Ra$ and the only
quantifier~$\forall$.

\medskip

\noindent Given a language of terms and propositions, a theory in deduction
modulo is defined not by a system of axioms, but by a system of
rewrite rules on terms and propositions.
Since in this paper we are not interested in the rewrite system
itself, but only in the congruence that it generates, we shall more
generally work with an arbitrary congruence over
propositions.

\medskip

\noindent Once the congruence over propositions has been fixed, the
principle of deduction modulo consists in adapting deduction rules
of natural deduction by allowing a proposition to be replaced by any
congruent proposition at each deduction step.

\subsection{Theories in minimal deduction modulo} \hfill\\

\noindent In this section, we present the definitions for the syntax of minimal deduction modulo (with proof-terms \`a la Curry {\em or} \`a la Church) and some basic properties.

\begin{definition}[Terms and propositions]\hfill\\
Given an infinite set of {\em term
  variables} (notation: $x$, $y$, $z$, etc.)\ as well as a first-order
signature defining a set of {\em function symbols} (notation $f$,
$f'$, etc.)\ and a non-empty set of {\em predicate symbols}%
\footnote{We need to assume that the set of predicate symbols is not
  empty to ensure that the language of propositions generated from the
  signature is not empty.}
(notation: $P$, $P'$, etc.), each function or predicate symbol being
given with a natural number called its {\em rank}, or {\em arity},
the formation rules for terms and propositions are the usual ones:
\begin{iteMize}{$\bullet$}
\item If $x$ is a variable, then~$x$ is a term.
\item If $f$ is a function symbol of rank $n$, and $t_1,\dots,t_n$
  are terms,\\
  then $f(t_1,\ldots,t_n)$ is a term.
\item If $P$ is a predicate symbol of rank $n$, and 
  $t_1,\dots,t_n$ are terms,\\
  then $P(t_1,\ldots,t_n)$ is an (atomic) proposition.
\item If $A$ and~$B$ are propositions, then so is $A\Ra B$.
\item If~$x$ is a variable and $A$ is a proposition,
  then $\forall x.A$ is a proposition.
\end{iteMize}
\end{definition}


\noindent As usual, propositions are considered modulo $\alpha$-conversion on term-variables.
We call {\em free} the term-variables of a proposition that are not bound by a universal quantification. 
The set of free term-variables of a proposition~$A$
is written $FV(A)$.
The operation of (capture avoiding) substitution is defined as usual,
and given terms~$t$, $u$ and a proposition~$A$, we denote by $(u/x)t$
(resp.\ $(u/x)A$) the term (resp.\ the proposition) obtained by
replacing every (free) occurrence of the variable~$x$ by the term~$u$ in
the term~$t$ (resp.\ in the proposition~$A$).
\begin{definition}[Congruence relation]\hfill\\
Given a first order signature, a theory is defined by a {\em congruence}~$\equiv$ over propositions 
\\ {\em i.e.} an equivalence relation such that for all term-variables $x$ and propositions $A$,$A'$,$B$,$B'$, 

  \begin{iteMize}{$\bullet$}
\item If $A\equiv A'$ and $B\equiv B'$, then
  $A\Ra B\equiv A'\Ra B'$.
\item If $A\equiv A'$, then $\forall x.A\equiv\forall x.A'$.
\end{iteMize}
\end{definition}

\smallskip

\noindent We now present two different systems of proof-terms for minimal
deduction modulo, depending on the amount of information we keep from
the derivation. Both are formed from an auxiliary set of
proof-variables, written $\alpha$, $\beta$, $\gamma$, etc.

\begin{definition}[Proof-terms \`a la Church]\hfill\\
The Church-style proof-terms system for minimal deduction modulo contains two forms
of $\lambda$-abstraction (one for introducing implication and one for
introducing universal quantification) and two forms of
application (for the corresponding eliminations).
Formally, the proof-terms (notation: $\pi$, $\pi'$, etc.)  are defined by:
$$\pi\quad::=\quad\alpha\quad|\quad
\lambda\alpha.\pi\quad|\quad\pi\pi
\quad|\quad\lambda x.\pi\quad|\quad\pi t
$$
Notice that the first form of $\lambda$-abstraction binds a proof-variable
whereas the second form binds a term variable.
\end{definition}

\begin{definition}[Proof-terms \`a la Curry]\hfill\\
The Curry-style proof-terms system  does not keep track anymore of the introductions and eliminations
of universal quantifications.
As a consequence, proof-terms are just pure $\lambda$-terms:
$$\pi\quad::=\quad\alpha\quad|\quad
\lambda\alpha.\pi\quad|\quad\pi\pi
$$
\end{definition}

\medskip

\noindent We call {\em neutral} those proof-terms which are not $\lambda$-abstractions. This notion of neutrality is fundamental in the definition of reducibility candidates.

\begin{definition}[Neutral proof-terms]\hfill\\
{\em Neutral} proof-terms are proof-terms of the form:
\begin{iteMize}{$\bullet$}
	\item $\alpha$, $\pi\pi'$ or $\pi~t$ in Church style,
	\item  $\alpha$ or $\pi\pi'$ in Curry style.
\end{iteMize}
\end{definition}

\medskip
\begin{definition}[Substitution]\hfill\\
 For both systems, the operation of (capture
avoiding) substitution is defined as expected.
In Church-system, we denote by $(t/x)\pi$ the proof-term
obtained by replacing every free occurrence of the first-order
variable~$x$ by the first-order term~$t$ in the proof-term~$\pi$.
In both Church-system and Curry-system, we denote by $(\pi'/\alpha)\pi$ the proof-term
obtained by replacing every free occurrence of the
proof-variable~$\alpha$ by the proof-term~$\pi'$ in the
proof-term~$\pi$.
\end{definition}


\medskip
\noindent We now define the type systems corresponding to those two proof-terms systems
(notice that we will use indifferently the words typing and deduction from now, via the proofs-as-programs paradigm).

\begin{definition}[Typing contexts]\hfill\\
A {\em typing context}, is a finite list of the form
$\quad\alpha_1:A_1,\ldots,\alpha_n:A_n\quad$
where $\alpha_1,\ldots,\alpha_n$ are pairwise distinct
proof-variables, and where~$A_1,\ldots,A_n$ are arbitrary
propositions. \linebreak
Given a typing context $\Gamma=\alpha_1:A_1,\ldots,\alpha_n:A_n$, we
write $FV(\Gamma)=FV(A_1)\cup\cdots\cup FV(A_n)$.
For all contexts~$\Gamma$ and~$\Gamma'$, we write
$\Gamma\subseteq\Gamma'$ when $(\alpha:A)\in\Gamma$ implies
$(\alpha:A)\in\Gamma'$ for all declarations $(\alpha:A)$.
\end{definition}

\noindent The typing rules for systems \`a la Church and \`a la Curry are given
in Fig.~\ref{f:Typing}.
These typing rules are the usual typing rules of natural deduction
(for minimal predicate logic) adapted to the framework of deduction
modulo, so that a proposition can always be replaced by a congruent
proposition at each step of the derivation.
For example, from a context containing a proposition $A$, the
\textsc{axiom} rule permits to derive any proposition~$A'$ that is
congruent to~$A$, and not only~$A$.

\begin{figure}[htp]
\renewcommand{\arraystretch} {3}
\begin{center}
\begin{tabular}{|c||c||p{70mm}|}
\hline
\rule[-1.5mm]{0mm}{5mm} 
	\textsc{axiom}  &  ~~Curry, Church~~ & 
	
	\ghost{8}
	$\irule{}
	{\Gamma,\alpha:A \vdash \alpha:B}
{	{\scriptsize \mbox{~~$A \equiv B$}}}$ \\
 
\hline \hline
\rule[-0mm]{0mm}{8mm} 
	~~$\Ra\textsc{-elim}$~~~~  &  Curry, Church & 
	\ghost{8}
	\irule{\Gamma \vdash \pi:C\hspace{0.5cm}\Gamma' \vdash \pi':A}
        {\Gamma \Gamma' \vdash (\pi~\pi'):B}
        {{\scriptsize \mbox{~~$C \equiv A \Ra B$}}}\\
        
 \hline \hline
\rule[-0mm]{0mm}{8mm} 
	$\Ra\textsc{-intro}$  &  Curry, Church & 
	\ghost{8}
	$\irule{\Gamma, \alpha:A \vdash \pi:B}
        {\Gamma \vdash \lambda \alpha.~\pi:C}
        {{\scriptsize \mbox{~~$C \equiv A \Ra B$}}}$\\

\hline \hline
\rule[-0mm]{0mm}{8mm} 
	$\forall\textsc{-elim}$  &  Church & 
	\ghost{8}
	$\irule{\Gamma \vdash \pi:B}
       					 {\Gamma \vdash \pi~t:C}
      					 {{\scriptsize \mbox{ ~~$B \equiv \forall x.A$,~~$C \equiv (t/x)A$}}}$\\

\cline{2-3}
\rule[-1mm]{0mm}{8mm} 
	  &  Curry & 
	\ghost{8}
	$\irule{\Gamma \vdash \pi:B}
       					 {\Gamma \vdash \pi:C}
      					 {{\scriptsize \mbox{ ~~$B \equiv \forall x.A$,~~$C \equiv (t/x)A$}}}$\\

\hline \hline
\rule[-0mm]{0mm}{8mm} 
	$\forall\textsc{-intro}$  &  Church & 
	\ghost{8}
	$\irule{\Gamma \vdash \pi:A}
       	{\Gamma \vdash \lambda x.\pi:B}
       	{{\scriptsize \mbox{ ~~$B \equiv \forall x.A,~~  x \not\in FV(\Gamma)$}}}$\\

\cline{2-3}
\rule[-1mm]{0mm}{8mm} 
	  &  Curry & 
	\ghost{8}
	$\irule{\Gamma \vdash \pi:A}
       	{\Gamma \vdash \pi : B}
       	{{\scriptsize \mbox{ ~~$B \equiv \forall x.A,~~  x \not\in FV(\Gamma)$}}}$\\

\hline 
\end{tabular}
\end{center}
\caption{Typing rules}
\label{f:Typing}
\end{figure}

\smallskip

\noindent We prove now some basic properties
of the two typing systems presented in Fig.~\ref{f:Typing}:
\begin{lemma}[Weakening]\label{l:Weakening}
  If $\Gamma\vdash\pi$ and
  $\Gamma\subseteq\Gamma'$, then $\Gamma\vdash\pi$.
\end{lemma}

\begin{proof}
  By induction on the derivation of~$\Gamma\vdash\pi$.
\end{proof}

\begin{lemma}[Substitutivity]\label{l:Substitutivity}~\vspace{-4mm}\\ 
  \begin{enumerate}[\em(1)]
  \item If $\Gamma_1,\alpha:A,\Gamma_2\vdash\pi:B$
    and $\Gamma_1\vdash\pi':A$, then
    $\Gamma_1,\Gamma_2\vdash(\pi'/\alpha)\pi:B$.
    \smallskip
  \item In Church-style system, if $\Gamma\vdash\pi:A$, then
    $(t/x)\Gamma\vdash(t/x)\pi:(t/x)A$.
    \smallskip
      \item In Curry-style system, if $~\Gamma\vdash\pi:A$, then
    $(t/x)\Gamma\vdash\pi:(t/x)A$.
   
  \end{enumerate}
  (where~$t$ is an arbitrary first-order term).
\end{lemma}

\begin{proof}
  Item~1 is proved by induction on the derivation of
  $\Gamma_1,\alpha:A,\Gamma_2\vdash\pi:B$ using
  Lemma~\ref{l:Weakening} in the case of the \textsc{axiom} rule.
  Items~2 and~3 are proved by induction on the derivation of
  $\Gamma\vdash\pi:A$.
\end{proof}


\bigskip


\subsection{$\beta$-reduction and strong normalization}
\label{normdef}\hfill\\

\noindent This section is devoted to the definitions of $\beta$-reduction and strong normalization for both systems.
$\beta$-reduction is a computation rule on proof-terms that simulates the elimination of cuts. And strong normalization is the property that ensures that all sequences of $\beta$-reductions from a proof are finite, leading to the property of cut elimination.

\begin{definition}[$\beta$-reduction]\hfill\\
  The relation of
  $\beta$-reduction is defined as the contextual closure of the
  following rules:
  \begin{iteMize}{$\bullet$}
  	\item $(\lambda\alpha.\pi)\pi' \ra (\pi'/\alpha)\pi~~$ and
    		$~~(\lambda x.\pi)t \ra (t/x)\pi~~$ in Church-style 
	\item $(\lambda\alpha.\pi)\pi' \ra (\pi'/\alpha)\pi~~$ in Curry-style.
\end{iteMize}
\end{definition}

\noindent We can see from this definition that $\beta$-reduction models both
$\Ra$-cuts and $\forall$-cuts in the system \`a la Church, while
it only models $\Ra$-cuts in the system \`a la Curry.
However, the strong
normalization of a proof-term corresponding to a purely logical
derivation of a sequent $\Gamma\vdash A$ does not depend \cite{CouMiq} on whether we
construct this proof-term \`a la Church or \`a la Curry when considering 
non-confusing theories (see section \ref{conf} for the definition).

\noindent As usual, we write $\pi\ra^*\pi'$ (resp.\ $\pi\ra^+\pi'$) if $\pi$
reduces to $\pi'$ in zero or more steps of $\beta$-reduction
(resp.\ in one or more steps of $\beta$-reduction).

\bigskip

\noindent Both systems defined above satisfy the
{\em subject reduction property} w.r.t.\ the corresponding notion of
$\beta$-reduction ({\em i.e.} $\beta$-reducing a term does not break well-typedness).

\begin{lemma}[Subject reduction]\label{p:sr}
  ~\\In both systems,
  if $\Gamma\vdash\pi:A$ and $\pi\ra\pi'$,
  then $\Gamma\vdash\pi':A$.
\end{lemma}

\proof
  By induction on the derivation of $\Gamma\vdash\pi:A$
  using Lemma~\ref{l:Substitutivity}.
 \qed

\medskip

\noindent We now formally define, for both systems \`a la Curry and \`a la Church, the {\em strong normalization} property: the fact that $\beta$-reduction terminates.

\begin{definition}
Given a proof-term $\pi$,
\begin{iteMize}{$-$}

\item[--] a {\em finite reduction sequence} starting from~$\pi$ is
  any finite sequence $(\pi_i)_{0\leq i \leq n}$ ($n\ge 0$) such that $\pi_0=\pi$
  and $\pi_i\ra\pi_{i+1}$ for all $0\le i<n$.
  The natural number~$n\ge 0$ is then called the {\em length} of the
  sequence $(\pi_i)_{0 \leq i \leq n}$;
\smallskip

\item[--] an {\em infinite reduction sequence} starting from~$\pi$ is any
  infinite sequence $(\pi_i)_{i\in\mathbb{N}}$ such that $\pi_0=\pi$
  and $\pi_i\ra\pi_{i+1}$ for all $i\in\mathbb{N}$.
\end{iteMize}
\end{definition}

\smallskip

\noindent The set of all finite reduction sequences starting from a given
proof-term~$\pi$ naturally forms a tree, which is called the
{\em reduction tree} of the term~$\pi$.
This reduction tree may be finite or infinite, but it is always
finitely branching since each proof-term has only
a finite numbers of $1$-reducts.

\begin{definition}[Normal forms]\hfill\\
  We say that a proof-term~$\pi$ is in {\em normal form} if there is no
  proof-term~$\pi'$ such that $\pi\ra\pi'$.
\end{definition}

\noindent Equivalently, a proof-term~$\pi$ is in normal form if and only if its
reduction tree is reduced to a singleton.

\begin{definition}[Strongly normalizing proof-terms]\hfill\\
  We say that a proof-term~$\pi$ is {\em strongly
    normalizing} if one of the following
  equivalent conditions holds:
  \begin{enumerate}[(1)]
  \item The reduction tree of~$\pi$ is finite.
  \item There is no infinite reduction sequence
    starting from~$\pi$.
  \end{enumerate}
  The set of all strongly normalizing terms is
  written $SN$.
\end{definition}

\medskip

\noindent  And we say that a theory in minimal deduction modulo is strongly normalizing if all its (well-typed) proof-terms are strongly normalizing.

\smallskip

\begin{definition}[Strongly normalizing theories]~\\
  We say that the theory is
  {\em strongly normalizing} if for all
  contexts~$\Gamma$, for all propositions $A$ and for all proof-terms
  $\pi$, $\Gamma\vdash\pi:A$ entails that~$\pi$ is
  strongly normalizing.
\end{definition}

\smallskip

\noindent Minimal deduction modulo allows to express both strongly normalizing and not strongly normalizing theories. For example, the theory defined by an empty signature and the trivial congruence relation (propositions are only congruent to themselves) is strongly normalizing since well-typed proof-terms of usual natural deduction do terminate. On the other side, whenever two propositions of the form $A$ and $A \Ra B$ are congruent in a theory, this theory is not strongly normalizing. Indeed, in that case, we have:
\vspace{1ex}
$$\irule{\alpha:A \vdash \alpha : A  \hspace{3ex} \alpha:A \vdash \alpha : A\Ra B}
{  \irule{\alpha:A\vdash \alpha\alpha : B}
            {\alpha:A\vdash \lambda \alpha. \alpha\alpha :  A \Ra B}
            {\hspace{2ex}\Ra\textrm{-intro}}
}      
{\hspace{2ex}\Ra\textrm{-elim}}$$

\vspace{.5ex}
\noindent And we also have $\alpha:A\vdash \lambda \alpha. \alpha\alpha :  A$ since $A \equiv A \Ra B$.
Hence $\alpha:A \vdash (\lambda \alpha. \alpha\alpha)(\lambda \alpha. \alpha\alpha) : B$, using the $\Ra$-elimination rule. But the proof-term $(\lambda \alpha. \alpha\alpha)(\lambda \alpha. \alpha\alpha)$ is not strongly normalizing since it reduces to itself in one step of $\beta$-reduction.



\subsection{On confusing and non-confusing theories}
\label{conf}\hfill\\
Now that we have defined how to express theories in minimal deduction 
modulo \`a la Curry and \`a la Church, let us focus on the particularly interesting
property of confusion, concerning the congruence relation defining some theory.

\begin{definition}[Confusion]
\label{confusion}\hfill\\
A congruence relation~$\equiv$ is said to be {\em confusing}, if it identifies two
non-atomic propositions starting with a different top-level connective or
quantifier. In minimal deduction modulo, this means that there exists propositions
$A,B,C$ such that $A \Ra B \equiv \forall x.C$.
\end{definition}

\noindent An example of rewrite rule leading to a confusing theory is $\forall x.(A \Ra B) \lra A \Ra \forall x.B$.
This sort of rewrite rule exhibits a major difference between expressing theories
in minimal deduction modulo \`a la Church and \`a la Curry.
Indeed, given a theory, adding that rewrite rule to the rewrite system defining $\equiv$,
 does change provability in Church-system whereas it does not in Curry-system (if $x$ is not free in $A$).
 
 \begin{iteMize}{$\bullet$}
\item In Church system, this rule is not admissible: if the rewrite system is empty, $\lambda \alpha.\pi$ (with $\pi$ a proof of $\forall x.B$ when $\alpha$ is a proof of $A$) is a proof of $A \Ra \forall x.B$ and it cannot be a proof of of a universally quantified proposition, in particular $\forall x.(A \Ra B)$. 
\smallskip
\item But in Curry system, all proofs of  $A \Ra \forall x.B$ are proofs of $A \Ra B$ (since all proofs of $\forall x.B$ are proofs of $(x/x)B$), and, by the same reasoning, they are also proofs of $\forall x.(A \Ra B)$.
 \end{iteMize}
 \bigskip
 
 \noindent With Alexandre Miquel, we explored, in \cite{CouMiq}, the relation between 
 strong normalization of a theory, when expressed \`a la Curry or \`a la Church. We proved that
 in the case of a non-confusing theory, strong normalization is equivalent in systems \`a la Church and \`a la
 Curry. And we also proved that in all cases strong normalization \`a la Curry entails strong normalization \`a la
 Church for a given theory. 
  
  \medskip
  
  \begin{theorem}[Church and Curry strong normalization \cite{CouMiq}]\hfill\\
  \label{ccsn}
  Let us consider a theory in minimal deduction modulo,
  \begin{iteMize}{$\bullet$}
  	\item if it is non-confusing then it is strongly normalizing in Curry style if and only if it is strongly normalizing in Church style.
  	\item in all cases, if it is strongly normalizing in Curry style then so is it in Church style.
  \end{iteMize}
  \end{theorem}
  
  \medskip
  
\noindent  Notice that the author conjectures that for confusing theories, strong normalization is also equivalent in systems \`a la Church and \`a la Curry. See the discussion in the conclusion.
  
%
%
%
%
%

\bigskip
\bigskip
\section{Reducibility candidates and pre-Heyting algebras}
\label{rcpha}

\noindent As explained before, strong normalization is an essential property for logic systems. Girard extended Tait's convertibility method to define reducibility candidates, which provides a general method to prove strong normalization.

\subsection{Concept of reducibility candidates}
\label{concept}\hfill\\
This section is devoted to introduce the concept of reducibility candidates, by explaining the key point of such proofs of strong normalization, in order to understand how reducibility candidates can be modified to get completeness while keeping soundness (for strong normalization).

\medskip

\noindent The main idea of reducibility candidates is to associate to each proposition $A$, a set of proof-terms called ${\cal R}_A$ containing only strongly normalizing proof-terms and then prove the  {\em adequacy lemma}, which entails strong normalization of the considered logical framework. 

\medskip

\noindent The {\bf {\em adequacy lemma}} states that if $\pi$ is a proof of $A$ (in a context $\Gamma$) then it belongs to ${\cal R}_A$ and is therefore strongly normalizing. 

\medskip

\noindent The proof of the adequacy lemma is done by induction on the length of the typing derivation $\Gamma \vdash \pi : A$, by case on the last rule used in this derivation. Therefore reducibility candidates have to be modeled on typing rules as we see in the following. Let us describe how to define those reducibility candidates in the case of the simply-typed $\lambda$-calculus (i.e. minimal predicate logic without universal quantification) in order to prove this adequacy lemma. We reason by case on the last rule used in the typing derivation $\Gamma \vdash \pi : A$ which can be one of the three following typing rules:
				
$$\irule{\alpha:A \in \Gamma} 
        {\Gamma \vdash \alpha:A}
        {\mbox{\small axiom}}\hspace{2cm}
        \irule{\Gamma, \alpha:A \vdash \pi: B} 
        {\Gamma \vdash \lambda \alpha.\pi:A \Ra B} 
        {\mbox{\small $\Ra$-intro}}\hspace{3cm} 
\irule{\Gamma \vdash \pi:A\Ra B\hspace{5mm}\Gamma \vdash \pi':A} 
        {\Gamma \vdash (\pi~\pi'):B} 
        {\mbox{\small $\Ra$-elim}}\hspace{2cm}$$ 
	
\vspace{2mm}
		\noindent To handle the $\Ra$-intro case, we actually need a more precise formulation of this lemma as: if $\Gamma \vdash \pi : A$ then if $\sigma$ is a  substitution such that for all variables $\alpha$  declared proof of $B$ in $\Gamma$, $\sigma \alpha \in {\cal R}_B$, then $\sigma \pi \in {\cal R}_A$  (notice that in this case, if such a substitution $\sigma$ exists, then $\pi \in SN$ since $\sigma \pi \in SN$). For convenience, we call {\em adequate} such substitutions in the following.
		
\medskip

\noindent Let us describe the scheme of the proof of this adequacy lemma, in order to bring up the different properties of reducibility candidates. Remind that we reason by case on the last rule used in the typing derivation $\Gamma \vdash \pi : A$ 

\medskip

\begin{iteMize}{$\bullet$}

 \item If the last rule used is \textit{axiom}, we conclude by hypothesis \mbox{on $\sigma$. }

\medskip

 \item  If the last rule used is \textit{$\Ra$-elim}, then  $\pi$ is of the form $\pi_1\pi_2$ and there exists a proposition $B$ such that $\Gamma \vdash \pi_1: B \Ra A$ and $\Gamma \vdash \pi_2 : B$ (both with shorter derivations). Let $\sigma$ be an adequate substitution (for $\Gamma$), we know by induction hypothesis that $\sigma \pi_1 \in {\cal R}_{B\Ra A}$ and  $\sigma \pi_2 \in {\cal R}_{B}$. To conclude, it is therefore convenient to make another assumption about those reducibility candidates: the fact that for all propositions $C$ and $D$, the set ${\cal R}_{C\Ra D}$ contains exactly all proof-terms which lead all elements of ${\cal R}_{C}$ to elements of ${\cal R}_{D}$. 

\medskip

\item The last case is a little bit trickier: if $\pi = \lambda \alpha.\mu$ and $A = B \Ra C$, then we know by induction hypothesis that for all $\nu \in {\cal R}_{B}$, $\sigma (\nu/\alpha)\mu \in {\cal R}_{C}$ (since $\sigma (\nu/\alpha)$ is adequate). But in order to conclude, we need to prove that its $\beta$-expansion, the neutral proof-term $\sigma((\lambda \alpha.\mu)~\nu)$, is also in ${\cal R}_{C}$. This leads to make another assumption on the sets ${\cal R}_D$, for all propositions $D$: we suppose that if a proof-term is neutral and all its $\beta$-reducts are in ${\cal R}_D$ then it is also in ${\cal R}_D$. Notice that we make an assumption on {\em all} $\beta$-reducts of a neutral term, not on only one $\beta$-reduct, since we want this property to be compatible with strong normalization. In our case, since $\beta$-reductions can also appear in $\lambda \alpha.\mu$ and in $\nu$, we can conclude by making a last assumption on the sets ${\cal R}_D$, for all propositions $D$: they are stable by $\beta$-reduction.
\end{iteMize}
		
\medskip
		
	\noindent	In sum, we define as reducibility candidates for the simply typed $\lambda$-calculus, the functions ${\cal R}_.$  from propositions to sets of proof-terms such that: 
	
	\medskip
	
	\noindent -- For all propositions $A$,  ${\cal R}_A$ satisfies the so-called (CR$_1$), (CR$_2$) and (CR$_3$) properties:

\medskip
		 (CR$_1$)~ ${\cal R}_A \subseteq SN$

		  (CR$_2$)~ if $\pi \in {\cal R}_A$ and $\pi \ra \pi'$ then $\pi' \in {\cal R}_A$

		  (CR$_3$)~ if $\pi$ is neutral and ${\cal R}_A$ contains all its one-step $\beta$-reducts, 
			then $\pi$  belongs to ${\cal R}_A$
		
\bigskip
		
		\noindent -- For all propositions $A$ and $B$, 
		\smallskip

		${\cal R}_{A \Ra B} = \{ \pi$ such that for all $\mu \in {\cal R}_A,~\pi\mu\in {\cal R}_B\}$.
		
		\bigskip
		\noindent Since we are able to define such a set of reducibility candidates for all propositions (by associating the set $SN$ to atomic propositions and using the property above as an inductive definition for the other propositions), we can conclude, via the adequacy lemma, that the simply-typed $\lambda$-calculus is strongly normalizing.
		
		\medskip
		
		\noindent Notice finally that reducibility candidates cannot be empty because of the (CR$_3$) property (non-emptiness is needed to build an adequate substitution for the contexts considered in the adequacy lemma). All normal neutral proof-terms, such as variables, have no $\beta$-reduct and are therefore in all reducibility candidates. 


\medskip

\subsection{Soundness for strong normalization}\hfill\\
To understand why this notion of reducibility candidates can be seen as a sound criterion for strong normalization, let us consider a very simple logical framework: the simply-typed $\lambda$-calculus modulo.
Applying the concept of deduction modulo, the simply typed $\lambda$-calculus can be extended by considering a congruence $\equiv$  on propositions, and authorizing to identify $\equiv$-equivalent propositions in typing derivations. This leads to consider the following adapted typing rules: 

\vspace{0mm}	
		
$$\irule{\alpha:A \in \Gamma} 
        {\Gamma \vdash \alpha:B}
        {\mbox{\small ~$A\equiv B$}}\hspace{2cm}
        \irule{\Gamma, \alpha:A \vdash \pi: B} 
        {\Gamma \vdash \lambda \alpha.\pi:C} 
        {\mbox{\small ~$C\equiv A \Ra B$}}\hspace{3.5cm} 
\irule{\Gamma \vdash \pi:C\hspace{5mm}\Gamma \vdash \pi':A} 
        {\Gamma \vdash (\pi~\pi'):B} 
        {\mbox{\small ~$C\equiv A \Ra B$}}\hspace{2cm}$$

\vspace{2mm}
		\noindent
In this logical framework, we can express strongly normalizing and non-strongly normalizing theories. For example, if we consider an atomic proposition $A$, the congruence generated by the rewrite rule $A \ra A$ expresses a strongly normalizing theory (remind that this rewrite rule concerns propositions and not proof-terms).  Whereas the congruence generated by the rewrite rule $A \ra A \Ra A$ expresses a non strongly normalizing theory (since in this case, the non-normalizing proof-term $(\lambda \alpha.\alpha\alpha)(\lambda \alpha.\alpha\alpha)$ is a proof of $A$). 

\medskip

\noindent In order to continue to be modeled on typing, since $\equiv$-congruent propositions are identified in typing rules, reducibility candidates for this logical framework have to satisfy another property: the fact that if $A$ and $B$ are two $\equiv$-congruent propositions then ${\cal R}_A = {\cal R}_B$. This is the idea of {\em pre-models} \cite{DW} which are the extension of the notion of reducibility candidates to deduction modulo. Provided this additional property, the proof of the adequacy lemma can be directly transposed to the simply typed $\lambda$-calculus modulo. 
Given a theory, the ability to build such a set of reducibility candidates (pre-model) via the method presented in the section \ref{concept} is lost in general because of this last property. And the existence of a pre-model provides a sound criterion for strong normalization of theories expressed in this logical framework.

%
%
%
%
%
%
\subsection{Semantic definition of reducibility candidates \`a la Church}\hfill\\

\noindent In \cite{DOW-TVA}, Gilles Dowek gave a semantic definition of this notion of pre-models, by defining
the notion of pre-Heyting algebra (also kwown as truth values algebra), on which pre-models can be defined as models.
We only define here the restriction of pre-Heyting algebras to minimal deduction modulo \`a la Church. 
 
\begin{definition}[pre-Heyting algebra]\hfill\\
Let $\calB$ be a set, $\leq$ be a relation on $\calB$, 
$\calA$ be a subset of $\wp(\calB)$,
 $\tildeimp$ be a function from $\calB \times \calB$ to
$\calB$ and $\tildefa$ be a function from $\calA$ to $\calB$,
the structure 
$\calB = \langle \calB, {\leq}, \calA,
\tildeimp, \tildefa\rangle$ is said to
be a {\em pre-Heyting algebra} if 

\begin{iteMize}{$\bullet$}
\item the relation $\leq$ is a pre-order,
\item for all $a \in \calB$ and $A \in \calA$, $a~\tildeimp~A$ is in $\calA$,
\item $\tildefa$ is an infinite greatest lower bound for $\leq$,
\end{iteMize}

\noindent (For $A \in {\cal A}$ and $a \in {\cal B}$, we write $a~\tildeimp~A$ for the set $\{a~\tildeimp~b$, for $b \in A\}$.) \vspace{2mm}
\end{definition}

\medskip
\noindent Let us now define the notion of model valued on a pre-Heyting algebra.

\begin{definition}[${\cal B}$-valued structure]\hfill\\
Let ${\cal L} = \langle f_i, P_j \rangle$ be a first-order signature
 and ${\cal B}$ be a pre-Heyting algebra, a {\em ${\cal
    B}$-valued structure} $ \calM  =
\langle  M , {\cal B}, \hat{f}_i, \hat{P}_j \rangle$  for the first order signature ${\cal L}$, is a structure
such that each $\hat{f_i}$ is a function from $ M ^n$ to $ M $ where $n$ is
the arity of the function symbol $f_i$ and each $\hat{P_j}$ is a function from $ M
^n$ to ${\cal B}$ where $n$ is the arity of the predicate symbol $P_j$. (We may call $M$ the {\em term-model} in the following.)
\end{definition}

\medskip

\begin{definition}[Environment]\hfill\\
  Given a set ${\calB}$-valued structure $ \calM  =
\langle  M , {\cal B}, \hat{f}_i, \hat{P}_j \rangle$ , an {\em environment} is a function which leads each term-variable to an element of $M$.
\end{definition}

\medskip

%

\begin{definition}[Interpretation] \label{def:interpretation}\hfill\\
Let ${\cal B}$ be a pre-Heyting algebra, $\calM$ be a ${\cal
  B}$-valued structure and $\phi$ be an environment. The interpretations 
  $\llbracket t \rrbracket_\phi^\calM$ of a term $t$ in $\calM$  and 
$\llbracket A \rrbracket_\phi^\calM$ of a proposition $A$ in $\calM$ are defined
as follows

\begin{iteMize}{$\bullet$}
\item $\llbracket x \rrbracket_{\phi}^\calM = \phi(x)$, 
\item $\llbracket f(t_1, ..., t_n) \rrbracket_{\phi}^\calM = 
\hat{f}(\llbracket t_1\rrbracket_{\phi}^\calM, ..., \llbracket
t_n\rrbracket_{\phi}^\calM)$,
\item $\llbracket P(t_1, ..., t_n) \rrbracket_{\phi}^\calM = 
\hat{P}(\llbracket t_1\rrbracket_{\phi}^\calM, ..., \llbracket
t_n\rrbracket_{\phi}^\calM)$,
\item $\llbracket A \Rightarrow B \rrbracket_{\phi}^\calM = 
\llbracket A \rrbracket_{\phi}^\calM 
~\tildeimp~
\llbracket B \rrbracket_{\phi}^\calM$, 
\item $\llbracket \fa x~A \rrbracket_{\phi}^\calM = 
\tildefa~\{
\llbracket A \rrbracket_{\phi + \langle x, e\rangle}~|~e \in  M \}$
when it is defined.
\end{iteMize}
\end{definition}

\medskip

\begin{remark} We omit $\calM$ from
  $\llbracket A \rrbracket_\phi^\calM$ when it is clear from
  context. 
  \\In all the pre-Heyting Algebras we consider in this paper,
  $\calA$ at least contains all the sets of the form $\{
  \llbracket A \rrbracket_{\phi + \langle x, e\rangle}~|~e \in  M \}$
  so that $\llbracket \forall x.A \rrbracket_\phi$ is always defined.
\end{remark}

\medskip

\noindent The following lemma comes for free with the previous definition that builds the interpretation of a proposition, inductively from the first order signature, given a $\calB$-valued structure. 

\begin{lemma} \label{lsub} For all propositions $A$ 
 and environments $\phi$,
$~~
\llbracket (t/x)A \rrbracket_{\phi} = 
\llbracket A \rrbracket_{\phi + \langle x, \llbracket t \rrbracket_\phi \rangle}
$. 
\end{lemma}
\begin{proof} By structural induction on $A$ and $u$.
\end{proof}

\medskip

\noindent This lemma shows a fundamental property concerning interpretations of term-substituted propositions. It is important to notice that in section \ref{tc}, when we propose another way to define interpretations of propositions, we shall need to add this property directly in our definition of models, whereas it is not the case presently.

\bigskip

\noindent A model is a ${\cal B}$-valued structure such that the associated interpretation identifies congruent propositions of the considered theory (it its original statement a model also identifies congruent terms, but this is of no interest for the present paper).

\begin{definition}[Model] \label{def:model}\hfill\\
The ${\cal B}$-valued structure $\calM$ is said to be {\em a model of}
a theory $({\cal L}, \equiv)$ if for all propositions ~$A$ and $B$ such that $A \equiv B$, 
 for all environments $\phi$, 
$~\llbracket A \rrbracket_\phi = \llbracket B \rrbracket_\phi$.
\end{definition}

\medskip

\noindent Finally we present the definition of the pre-Heyting algebra of reducibility candidates and state that the existence of a model valued on this algebra is a sound semantic criterion for strong normalization of theories expressed in minimal deduction modulo (notice that the original theorem concerns whole deduction modulo \`a la Church \cite{DOW-TVA}).

\begin{definition}[The algebra of reducibility candidates]\label{acr}\hfill\\
The domain of the algebra is $\calC$ the set of reducibility candidates (i.e. the set of sets of proof-terms which satisfy (CR$_1$), (CR$_2$) and (CR$_3$)). 
\\ The $\calC$-valued structure $ \calM  = \langle  M , {\calC}, \hat{f}_i, \hat{P}_j \rangle$ is composed of $M$ the set of terms,
each $\hat{f}_i$ is the function symbol $f_i$ itself and each $\hat{P}_j$ is the constant function \mbox{leading all tuples of terms to $SN$.}
\\ The set $\calA$ is $\wp(\calB)$. 
\\ $\leq$ is inclusion.
\\ For all $a,b \in \calC$, $a \tildeimp b$ is defined as the set of proof-terms $\pi$ such that for all $\mu \in a$, $\pi\mu \in b$.
\\ For all $A \in \calA$, $\tfa A$ is the set of proof-terms $\pi$ such that for all terms $t$ and $a \in A$, $\pi~t \in a$.
\end{definition}


\medskip

\begin{theorem}[Soundness \cite{DW,DOW-TVA}]\hfill\\
If a theory in minimal deduction modulo has a $\calC$-valued model then it is strongly normalizing.
 In other words, the existence of a $\calC$-valued model is a sound criterion for strong normalization of a theory in minimal deduction modulo (\`a la Church).
\end{theorem}

\begin{proof}
The proof of that theorem consists in proving the (right form of) {\em adequacy lemma} presented in section (see \ref{concept}). 
\end{proof}

\bigskip

\subsection{On the interpretation of universal quantification}\hfill\\

\noindent The previous definition, of reducibility candidates as $\calC$-valued models, leads to interpret a proposition $\forall x.A$  as: \vspace{1mm}
\\$\llbracket \fa x.A \rrbracket_{\phi}
= \tfa ~\{
\llbracket A \rrbracket_{\phi + \langle x, t\rangle}~|~t \in  M \}  \vspace{1mm}
\\$\mbox{~}\hspace{13mm}$= \{ \pi $ such that for all terms $t_1$, $t_2$,~ $\pi t_1 \in \llbracket A \rrbracket_{\phi + \langle x, t_2\rangle}\} \vspace{1mm}
\\$\mbox{~}\hspace{13mm}$= \{ \pi $ such that for all terms $t_1$, $t_2$,~ $\pi t_1 \in \llbracket (t_2/x)A \rrbracket_{\phi}\}$.

\bigskip

\noindent This definition does not capture exactly the $\fa$-elim rule. Indeed it is too restrictive since $t_1$ and $t_2$ are not synchronized. For example, if we consider the theory of natural numbers, we could imagine a proof of $\forall x,~ x \geq 0$, which does not gives a proof of $t_1 \geq 0$ when applied to $t_2$, a term different from $t_1$. This sort of proof would not belong to $\llbracket \forall x,~ x \geq 0\rrbracket_{\phi}$. This causes difficulties to prove completeness of $\calC$-valued models as such. In the following, we present two solutions. The first one consists in making more precise this interpretation but this necessitates to extend the notion of pre-Heyting algebra. And the second one considers proof-terms \`a la Curry.


\bigskip
\bigskip

\section{Toward completeness}
\label{tc}

\noindent We have seen that reducibility candidates provide  a sound criterion for strong normalization of theories expressed in minimal deduction modulo. In order to prove that it also forms a complete criterion, one has to prove that whenever a theory is strongly normalizing, it is possible to build a model valued on reducibility candidates ($\calC$) for that theory. The method we use is closer to Henkin's proof \cite{Henkin49} of completeness of Boole-valued models for consistency in first order logic, than to G\"odel's one \cite{Godel} since it consists in directly building that reducibility candidates model from the hypothesis of strong normalization.

\smallskip

\noindent How to use that hypothesis of strong normalization? A possibility is to consider well-typed proof-terms, {\em i.e.} first associate to each ordered pair of a proposition and an environment the set of proof-terms that are proofs of that proposition (in some context). And second, prove that it forms a reducibility candidate when the theory is strongly normalizing. 
We shall see that this na\"ive idea does not apply as such, but it brings up a new manner to define models valued on a pre-Heyting algebra.  The usual way, that was presented in section \ref{rcpha}, is to define the interpretation only on atomic propositions, and then obtain its value on non-atomic propositions by using $\tra$ and $\tfa$ as inductive definitions. This way, the only needed property to obtain a model is the fact that interpretations of congruent propositions are equal (we shall say {\em adapted to the congruence}). The other way we propose consists in defining more generally interpretations as functions from ordered pairs of a (not necessarily atomic) proposition and a environment to elements of the domain of the considered algebra. 

\smallskip

\noindent This way, models are defined as interpretations

\begin{iteMize}{(1)}
\item which are adapted to the congruence,
\item which satisfies the property of lemma \ref{lsub}, namely the {\em substitution property}, 
\item and such that the interpretation of $A \Ra B$ is the value of $\tra$ applied to the interpretations of $A$ and $B$, and the analog property for $\tfa$ (we shall say {\em adapted to the connectives}).
\end{iteMize}

\smallskip

\noindent As we shall see in the following,  defining an interpretation adapted to the congruence by definition and then proving that it is also adapted to the connectives, may be simpler than the opposite. It is the reason why we propose those slightly different definitions of interpretations and models, which emphasize that models can be built in a different way from the usual one.
Notice that we give a simplified definition since we shall only consider models based on a term-model equal to the set of terms defined by the considered first order signature.

\begin{definition}[Environment (2)]\hfill\\
  {\em Environments} are now functions from term-variables to terms ({\em i.e.} substitutions).
\end{definition}

\smallskip

\begin{definition}[Interpretation (2)]\hfill\\
  Given a pre-Heyting algebra $\calB$,
   a {\em $\calB$-valued interpretation} is a function which leads
all ordered pairs of a proposition and a environment to an element of $\calB$.
\end{definition}

\smallskip

\noindent Let $\calT$ be a theory expressed in minimal deduction modulo, given by a first order signature and a congruence relation $\equiv$.

\begin{definition}[Model (2)] \hfill\\
We write $M$ the set of terms of $\calT$.
\\Let ${\cal B} = \langle \calB, {\leq}, \calA,
\tildeimp, \tildefa\rangle$ be a pre-Heyting algebra.
\\ A $\calB$-valued interpretation (leading propositions $A$ and environments $\phi$ to  $\llbracket A \rrbracket_{\phi}^M$) is a {\em model} of the theory $\calT$ if and only if for all environments $\phi$, propositions $A,B$ terms $t$ and term-variables $x$,

\medskip

\noindent -- it is {adapted to the connectives}, \emph{i.e}
 
\begin{iteMize}{(1)}

\item $\llbracket A \Ra B \rrbracket_{\phi}^M = \llbracket A \rrbracket_{\phi}^M \tildeimp \llbracket B \rrbracket_{\phi}^M$

\item $\llbracket \fa x.A \rrbracket_{\phi}^M = 
\tildefa~\{
\llbracket A \rrbracket_{\phi + \langle x, t\rangle}~|~t \in  M \}$
when it is defined.
\end{iteMize}

\bigskip

\noindent -- it satisfies the substitution property, {\em i.e} $~\llbracket (t/x)A \rrbracket_{\phi}^M = \llbracket A \rrbracket_{\phi + \langle x,t \rangle}^M $

\medskip
\noindent -- it is adapted to the congruence, {\em i.e}  $~$if $A \equiv B$ then $\llbracket A \rrbracket_{\phi}^M \equiv \llbracket B \rrbracket_{\phi}^M$
\end{definition}

\smallskip

\noindent Let us now see how to tune the usual definition of reducibility candidates to obtain completeness for strong normalization. 

\subsection{On the (CR$_3$) property}
\label{otcr3p}\hfill\\
As previously mentioned, the way we use the strong normalization hypothesis is to consider first well-typed proof-terms, by interpreting a proposition by its proofs (in some context). Notice that such an interpretation is obviously adapted to the congruence since the sets of proofs of two equivalent propositions are equal. 
But such sets of well-typed proofs cannot satisfy (CR$_3$) since all sets of proof-terms satisfying (CR$_3$) contain ill-typed proof-terms. For example, $\alpha\alpha$ is neutral and normal hence it belongs to all reducibility candidates.
But if $\alpha\alpha$ is well-typed in a theory then so does $(\lambda\alpha.\alpha\alpha)(\lambda\alpha.\alpha\alpha)$ which is not normalizing and therefore cannot be well-typed in a strongly normalizing theory (as seen in the end of section \ref{normdef}). 
\\ Since we want to avoid ill-typed proof-terms like $\alpha\alpha$ from our new reducibility candidates, we make a first restriction on the (CR$_3$) property leading to, for some set $E$ of proof-terms:

\medskip
\noindent (CR$_{3\textrm{aux}}$) if a proof-term is neutral, {\em not normal} and all its one-step reducts belong to $E$ \\ \mbox{\hspace{9ex}} then it also belongs to $E$.

\smallskip

\noindent This way, (CR$_{3\textrm{aux}}$)-extensions of proofs of a proposition are proof-terms such that all reduction sequences from it eventually reach a proof of  that proposition.

\noindent But if we simply define the interpretation of a proposition as the (CR$_{3\textrm{aux}}$)-extension of proofs of that proposition, we do not get a interpretation adapted to the connective $\Ra$. Indeed, in that case, if $\pi$ belongs to the interpretation of a proposition $B$ (let us write it $\dd{B}$ without considering environments for the moment) and $\pi$ is not a proof of $B$ then $\pi$ is a (CR$_{3\textrm{aux}}$)-extension of  a proof of $B$, and for all proof-variables $\alpha$, not free in $\pi$, $\lambda\alpha.\pi$ belongs to $\dd{A} \tra \dd{B}$ but not to $\dd{A \Ra B}$ (with the usual $\tra$ of reducibility candidates). 

\smallskip

\begin{iteMize}{$\bullet$}
\item For all $\pi' \in \dd{A}$, one can prove that all reducts of the neutral proof-term $(\lambda\alpha.\pi)\pi'$ belong to $B$ by induction on the lengths of the maximal reductions sequence from $\pi$ and $\pi'$, and the fact that the head-reduct $(\pi'/\alpha)\pi = \pi \in \dd{B}$. Hence $\pi$ belongs to $\dd{A} \tra \dd{B}$.

\item But if $\pi$ is not a proof of $B$ then $\lambda \alpha.\pi$ is not a proof of $A \Ra B$ and (CR$_{3\textrm{aux}}$) cannot prove that $\lambda \alpha.\pi$ belongs to $\dd{A \Ra B}$ since $\lambda \alpha.\pi$ is not neutral. 
\end{iteMize}

\smallskip

\noindent Hence $\dd{A \Ra B} \neq \dd{A} \tra \dd{B}$, {\em i.e.} $\dd{.}$ is not adapted to the connectives.

\smallskip

\noindent In order to get connectives adaptation, we propose to relax this (CR$_{3\textrm{aux}}$) property by authorizing those "neutral not normal expansions" not only one by one at the root of the syntax tree representing  a proof-term, but simultaneously at different nodes of that tree. This leads to the following definition:

\medskip

\noindent (CR$_3'$)  for all $n \in \N$, for all proof-terms $\nu, \mu_1,\dots,\mu_n$,  if 
		
		\begin{iteMize}{-} 
			\item for all $i \leq n$, $\mu_i$ is neutral and not normal,
			\item for all proof-terms $\rho_1,\dots, \rho_n$ such that for all $i \leq n$, $\mu_{i} \ra \rho_i$, we have $[\rho_i/\alpha_i]_{i}\nu \in E$	
			\item{~} then we have $[\mu_i/\alpha_i]_i\nu \in E$.
		\end{iteMize}

\bigskip 

\noindent where $[\mu_i/\alpha_i]_i\nu$ denotes the sequence of substitutions with capture of $\alpha_i$ by $\mu_i$ for $0 \leq i \leq n$.

\bigskip

\noindent With this definition, if $\pi$ is a (CR$_3'$)-expansion of a proof of $B$ (in some context) and $\alpha$ is a proof-variable not free in $\pi$, then $\lambda \alpha.\pi$ is a (CR$_3'$)-expansion of a proof of $A \Ra B$ (in the same context) (see lemma \ref{lambdacl}). And we get back the fact that our interpretation is adapted to the connective $\Ra$.

\medskip

\subsection{On the interpretation of the universal quantification}\hfill\\
A last problem for proving completeness of usual reducibility candidates for minimal deduction modulo \`a la Church comes from the definition of $\tfa$.
As seen in the previous section, this definition leads to interpret a proposition $\forall x.A$  as the set of proof-terms which lead all terms $t_1$ to the interpretation of $(t_2/x)A$, for all terms $t_2$. 
\\ This prevents the interpretation defined above, based on well-typed proof-terms, from being adapted to the connective $\fa$, since the classical $\tfa$ does not model precisely enough the $\fa$-elim rule. Indeed if $\pi$ is a proof of $\forall x.A$, and $t_1$, $t_2$ are terms, the $\fa$-elim rule cannot help to deduce that $\pi t_1$ is a proof of $(t_2/x)A$.
\\ In order to synchronize those two terms in the definition of $\tfa$ in reducibility candidates \`a la Church, the author defined in \cite{pstt} and \cite{Theseoim} the notion of language-dependent truth values algebras (ldtva). Defining reducibility candidates (with (CR$_3'$)) as a model valued on a ldtva provides a sound and complete criterion for strong normalization of {\em non-confusing} theories in minimal deduction modulo \`a la Church.
\\ In the next section, we define sound and complete reducibility candidates for {\em both} confusing and non-confusing theories in minimal deduction \`a la Curry.
%


\bigskip
\bigskip

\section{Complete reducibility candidates \`a la Curry}
\label{crcalac}

\noindent This last section is devoted to the definition of a complete sound and complete semantics for strong normalization in minimal deduction modulo \`a la Curry. Considering Curry-style proof-terms allows to use the classical notion of pre-Heyting algebra. Since terms do not appear in proof-terms in minimal deduction modulo \`a la Curry, it allows to define $\tfa$ as a usual intersection. Moreover, it provides a stronger result than previous results (in minimal deduction modulo \`a la Church) since it concerns both confusing and non-confusing theories. And, icing on the cake, the proof of completeness of those reducibility candidates is considerably shorter than the one concerning minimal deduction modulo \`a la Church.

\medskip

\noindent The new pre-Heyting algebra of reducibility candidates we propose for minimal deduction \`a la Curry differs on two points, from the usual (only sound) one for minimal deduction \`a la Church, presented in section \ref{rcpha}. The first point is that $\tfa$ now is {\em classical set-intersection} since we consider proof-terms \`a la Curry. The second point is that the domain we now consider is the set of proof-terms satisfying (CR$_1$), (CR$_2$) and the new property {\em (CR$_3'$)} (to ensure completeness while keeping soundness for strong normalization).
 
\begin{definition}[The algebra of complete reducibility candidates \`a la Curry: $\cprime$]\label{accr}\hfill\\
The domain of $\cprime$ is the set of {\em non-empty} sets or proof-terms which satisfy the properties (CR$_1$), (CR$_2$) and (CR$_3'$). 
\\ The set $\calA$ is $\wp(\calC')$. 
\\ $\leq$ is set inclusion.
\\ For all $a,b \in \cprime$, $a \tildeimp b$ is the set of proof-terms $\pi$ such that for all $\mu \in a$, $\pi\mu \in b$.
\\ For all $A \in \calA$, $\tfa A$ is the set of proof-terms $\pi$ belonging to all $a \in A$.
\end{definition}

\noindent This definition provides a pre-Heyting algebra since we can easily check that for all $a,b\in \calC'$ and $A \in \wp(\calC')$, $a ~\tra~ b \in \calC'$ and $\tfa A \in \calC'$, and that $\tildefa$ is an infinite greatest lower bound for set inclusion. 

\subsection{Soundness}\hfill\\
In this section, we prove that the existence of a $\cprime$-valued model entails strong normalization of theories expressed in minimal deduction modulo. In other words, replacing the usual (CR$_3$) by (CR$_3'$) keeps soundness for strong normalization. Soundness (theorem \ref{sound}) is entailed, as usual, by the (right form of) adequacy lemma (lemma \ref{adq2}).

\begin{lemma}[Adequacy]
\label{adq2}\hfill\\
 If ~$\dd{.}_.$ is a $\cprime$-valued model of  a theory in minimal deduction modulo \`a la Curry,
	\\then for all propositions $A$, contexts $\Gamma$,environments $\phi$, proof-terms $\pi$ and substitutions $\sigma$ 
	such that for all declarations $\alpha: B$ in $\Gamma$, $\sigma \alpha \in \dd{B}_\phi$, we have:

	\begin{center} if $\Gamma \vdash \pi:A$ then $\sigma \pi \in \dd{A}_\phi$.\end{center}
\end{lemma}

\proof By induction on the length of the derivation of $\Gamma \vdash \pi: A$. By case on the last rule used. If the last rule used is :
	\begin{iteMize}{$\bullet$}
		\item axiom: in this case, $\pi$ is a variable $\alpha$, and $\Gamma$ contains a declaration $\alpha:B$ with $A \equiv B$.
			Then $\sigma  \alpha \in \dd{B}_\phi = \dd{A}_\phi$.
	\smallskip
		\item $\Ra$-intro: in this case, $\pi$ is an abstraction $\lambda\alpha.\tau$, and we have \mbox{$\Gamma,\alpha:B \vdash \tau:C$} 
				with $A \equiv B \Ra C$. Let $\sigma'$ such that for all variables $\beta$ declared in $\Gamma$, $\sigma' \beta  = \sigma \beta$
				and $\sigma' \alpha$ is an element of $\dd{B}_\phi$. Then $\sigma'\tau \in \dd{C}_\phi$ by induction hypothesis (and   
				$\sigma'\tau$ is in $SN$, therefore $\sigma (\lambda \alpha.\tau)$ is also in $SN$).  Let $\pi' \in \dd{B}_\phi$, we prove by induction 
				on the sum of maximal lengths of a reductions sequence from $\sigma(\lambda\alpha.\tau)$ and $\pi'$ (each in $SN$) 
				that every one-step reduct of the neutral not normal proof-term $\sigma(\lambda\alpha.\tau) ~ \pi'$ is in $\dd{C}_\phi$.
				If the one-step reduct is $\sigma(\pi'/\alpha)\tau$, we conclude by induction hypothesis (on the length of the derivation)
				since $\pi' \in \dd{B}_\phi$. Otherwise, the reduction takes place either in $\sigma(\lambda\alpha.\tau)$, either in $\pi'$. 
				We conclude first by induction hypothesis on the sum of the maximal lengths of reductions sequence from $\sigma(\lambda\alpha.\tau)$
				and $\pi'$. And second by the fact that both $\dd{B}_\phi $ and $\dd{B\Ra C}_\phi $ satisfy (CR$_2$).
				Finally, $\sigma(\lambda\alpha.\tau) ~ \pi' \in \dd{C}_\phi$, since it satisfies (CR$_3'$) and 
					$\sigma(\lambda\alpha.\tau) ~ \pi'$ is neutral, not normal.
				Hence  $\sigma(\lambda\alpha.\tau) \in \dd{B}_\phi \tra \dd{C}_\phi = \dd{B\Ra C}_\phi = \dd{A}_\phi$
		\smallskip 
		\item $\Ra$-elim: in this case, $\pi$ is an application $\rho \tau$, and we have \mbox{$\Gamma \vdash  \rho : C \equiv B\Ra A$} and 
		$\Gamma \vdash \tau : B$. Therefore, by induction hypothesis, 
		$\sigma\rho\in \dd{B\Ra A}_\phi = \dd{B}_\phi \tra \dd{A}_\phi$ and  $ \sigma \tau \in \dd{B}_\phi$. 
		Therefore $\sigma (\rho \tau) \in \dd{A}_\phi$.
	\smallskip 
		\item $\forall$-intro: in this case, we have \mbox{$\Gamma \vdash \pi: B$} with 
			\mbox{$A \equiv \forall x.B$}.
			Hence for all terms $t$,
			$\sigma \pi \in \dd{B}_{\phi+\langle x,t \rangle}$  by induction hypothesis,
			since, $\phi+\langle x,t \rangle$ is an environment.
			And $\sigma \pi \in \dd{\forall x.B}_\phi = \dd{A}_\phi$ by definition of $\tfa$. 	
			
	\smallskip
		\item $\forall$-elim: in this case, we have \mbox{$\Gamma \vdash \pi: \forall x.B$} with 
			\mbox{$A \equiv (t/x)B$}. Hence, by induction hypothesis, $\sigma \pi \in \dd{B}_{\phi + \langle x,t \rangle} = \dd{(t/x)B}_\phi = \dd{A}_\phi$, by the substitution property.\qed
	\end{iteMize}\smallskip

\noindent As previously mentioned, the adequacy lemma directly entails soundness for strong normalization.

\begin{theorem}[Soundness]
	\label{sound}\hfill\\
	If a theory in minimal deduction \`a la Curry has a $\cprime$-valued model, then it is strongly normalizing.
\end{theorem}

\begin{proof}
	If $\dd{.}_.$ is a $\cprime$-valued model of this theory then for all judgements $\Gamma \vdash \pi : A$ and $\sigma$ and $\phi$ 
	as in the previous proposition, we have $\sigma\pi \in \dd{A}_\phi$ hence $\sigma \pi \in SN$, therefore $\pi \in SN$.
\end{proof}

\smallskip

\subsection{Completeness}\hfill\\
In this section, we prove that the definition of $\cprime$-valued model also gives a complete criterion for strong normalization in minimal deduction \`a la Curry. Following Henkin's method \cite{Henkin} rather than G\"odel's one \cite{Godel}, we build directly a $\cprime$-valued model from the strong normalization hypothesis of a theory.
As explained in section \ref{tc}, we interpret propositions $A$  as the (CR$_3'$) expansion of proofs of $A$  in a general context $\Delta$.
We prove (lemma \ref{clc'}) that this interpretation takes its values in $\cprime$ {\em when the theory is strongly normalizing} (the strong normalization hypothesis is only needed to prove that this interpretation satisfy (CR$_1$)). This interpretation constitutes a model since it is trivially adapted to the congruence, it is adapted to the substitution (lemma \ref{clsubst}), and to the connectives (lemmas \ref{clramorph} and \ref{clfamorph}). Finally we obtain the completeness theorem \ref{comp}.

\begin{definition}[The universal context $\Delta$]\hfill\\
 	We consider a context which contains an infinite number of declarations for each proposition of the considered theory. 
\end{definition}

\begin{definition}[$\Omega$, a particular set of proof-terms]\hfill\\
 	For convenience, we write $\Omega$ the set of strongly normalizing, neutral, not normal proof-terms.  
\end{definition}

\noindent In the following, when $i$ and $n$ are integers, $\mu_1,\dots,\mu_n$, $\rho_1,\dots,\rho_n$ are proof-terms, $\alpha_1,\dots,\alpha_n$ are proof-variables, we shall write $[\mu_i/\alpha_i]_{i\leq n}$ for the substitution (with capture) \linebreak $[\mu_n/\alpha_n]\dots[\mu_1/\alpha_1]$ (we may write $[\mu_i/\alpha_i]_{i}$ when $n$ is clear from context). 
We shall also write $(\mu_i)_i \ra (\rho_i)_i$, when for all $i \leq n$, $\mu_i \ra \rho_i$. 

\bigskip

\noindent We define the interpretation of a proposition $A$ and an environment $\phi$ ({\em i.e.} a term-substitution in our case), as the (CR$_3'$)-countable iteration expansion of the set of proofs of $\phi A$.

\begin{definition}[Closure]\hfill\\
	For all propositions $A$ and environments $\phi$, we define $Cl(A)_\phi$ as follows : \\
	for all $k \in \N$,
	\begin{iteMize}{$\bullet$}
		\item $Cl^0(A)_\phi = \{ \pi$  such that $\Delta \vdash \pi : \phi A\}$ \vspace{2.5mm}

		\item 
			$Cl^{k+1}(A)_\phi =  \{\pi$ such that there exists $n \in \N$, a proof-term $\nu_\pi$ and $ (\mu_i)_{i\leq n}  \subseteq \Omega$:\\
			 \mbox{~}\hspace{22mm} such that $\pi = [\mu_i/\alpha_i]_{i\leq n}~\nu_\pi$ and for all $(\rho_i)_{i\leq n}$\\
			 \mbox{~}\hspace{22mm} if $(\mu_i)_i \ra (\rho_i)_i$ then $[\rho_i/\alpha_i]_{i\leq n}~\nu_\pi \in Cl^k(A)_\phi \}$\vspace{2.5mm}


		\item $Cl(A)_\phi =  \mathop{\cup}_{j \in \N} Cl^j(A)_\phi$
	\end{iteMize}
\end{definition}

\noindent Notice first that the strong normalization of the $\mu_i$ is not a necessary hypothesis but it simplifies 
item (CR$_1$) of proof of lemma \ref{clc'} (which is detailed in \cite{Theseoim}). Notice also that this (CR$_3'$)-expansions iteration is monotonous for inclusion, {\em i.e.} for all propositions $A$, environments $\phi$ and $k\in \N$, $Cl^k(A)_\phi \subseteq Cl^{k+1}(A)_\phi$.
Hence for all propositions $A$ and environments $\phi$, $Cl(A)_\phi$ is not empty since neither is $Cl^0(A)_\phi$ (it contains, in particular, all proof-variables declared of type $\phi A$ in $\Delta$).


\bigskip

\noindent We prove now that this interpretation leads all ordered pairs of a proposition and an environment to elements of $\cprime$ {\em when the considered theory is strongly normalizing}.

\begin{lemma} 
\label{clc'}~\\
	If the considered theory is strongly normalizing then for all propositions $A$ and environments $\phi$, $Cl(A)_\phi$ belongs to $ \cprime$.
\end{lemma}

\proof Let $A$ be a proposition and $\phi$ be an environment.
	\begin{iteMize}{(CR$_{2}$)}
		\item [(CR$_{2}$)]  Let $\pi \in Cl(A)_\phi$ and $\pi'$ a proof-term such that $\pi \ra \pi'$. 
			  Then there exists (a minimal) $k \in \N$ such that \mbox{$\pi \in Cl^k(A)_\phi$}.
			 By induction on $k$.
		\begin{iteMize}{-}
			\item If $k=0$, then $\Delta \vdash \pi :  \phi A$, therefore $\Delta \vdash \pi' :  \phi A$ by subject-reduction.
			\item If $k>0$,  then $\pi = [\mu_i/\alpha_i]_{i}\nu$ with each $\mu_i$ in $\Omega$, 
				and such that for all $(\rho_i)_{i}$  with 
				$(\mu_i)_i \ra (\rho_{i})_i$, we have $[\rho_i/\alpha_i]_{i}\nu \in Cl^{k-1}(A)_\phi$. 
				 Since each $\mu_i$ is neutral, the redex we reduce in pi is either in some $\mu_i$ or in $\nu$. 				  
				Thus,
				\smallskip
				
				\begin{iteMize}{--}
					\item Either $\pi' = [\rho_{i_0}/\alpha_{i_0}][\mu_i/\alpha_i]_{i \neq i_0}\nu$, with 
						$\mu_{i_0} \ra \rho_{i_0}$. In this case, $\pi'$ belongs to $Cl(A)_\phi$ (by considering the substitution $[\mu_i/\alpha_i]_{i \neq i_0}$ 
						 on the proof-term $[\rho_{i_0}/\alpha_{i_0}]\nu$).
						 \smallskip
					\item Or $\pi' = [\mu_i/\alpha_i]_{i}\nu$ with $\nu \ra \nu'$.Hence $[\rho_i/\alpha_i]_i\nu’ \in Cl^{k-1}(A)_\phi$ since
						$Cl^{k-1}(A)_\phi$ satisfies (CR$_2$) by induction hypothesis on $k$.
						And we conclude that $\pi' = [\mu_i/\alpha_i]_i\nu’ \in Cl^{k}(A)_\phi$.
				\end{iteMize}
				
		\end{iteMize}
		\smallskip
		\item [(CR$_{1}$)]  The fact that $Cl(A)_\phi$ only contains strongly normalizing proof-terms is not so hard to show but the proof (see \cite{Theseoim}) is still quite technical and long. The proof scheme is globally the same as the classical standardization theorem proof by Curry and Feys \cite{CurryCL}. It uses parallel reduction and defines, as in Curry and Feys proof, two kinds of reductions for a proof-term $[\mu_i/\alpha_i]_i\nu$: reductions in $\nu$ versus reductions in the $\mu_i$ (Curry and Feys distinguish head reductions from other reductions in their original proof). Notice finally that (CR$_1$) must be proved after (CR$_2$) since the latter is used in the proof of the former.

\smallskip
		\item [(CR$_3'$)] $Cl(A)_\phi$ satisfies (CR$_3'$) by construction, using the (CR$_1$) property and the fact that the (CR$_3'$) extension $[\mu_i/\alpha_i]_i\nu$ of a proof-term $[\rho_i/\alpha_i]_i\nu$ in SN is necessarily in SN hence so are the $\mu_i$ (see  \cite{Theseoim} for a detailed proof).\qed

	\end{iteMize}

\smallskip

\noindent Let us now prove an important property that is necessary to prove that our interpretation is adapted to both $\Ra$ and $\tfa$.

\begin{lemma}
\label{mink}\hfill\\
For all propositions $A$, environments $\phi$, and $\pi \in Cl(A)_\phi$, there exists $k \in \N$, less than or equal 
to the maximal length of a reductions sequence from $\pi$ such that $\pi \in Cl^k(A)_\phi$.
\linebreak In particular, if $\pi$ is in normal form then $\pi \in Cl^0(A)_\phi$ i.e. $\Delta \vdash \pi : \phi A$.
\end{lemma}

\begin{proof}
	By induction on the maximal length $m$ of a reductions sequence from $\pi$. If $m=0$ then $\pi$ is in normal form so it is necessarily in $Cl^0(A)_\phi$, since the proof-terms that belong to $Cl^{k+1}(A)_\phi$ but not to $Cl^k(A)_\phi$ cannot be in normal form. If $m>0$, let $\mu$ be some subterm of $\pi$ that is a redex. 
	$\mu$ is neutral, not normal and strongly normalizing (since $\pi \in SN$).
	Let $\pi'$ be a proof-term obtained by reducing a redex of $\pi$ that is also a redex of $\mu$, and $\rho$ be the proof-term obtained by reducing that same redex in $\mu$. 
	 Let us write $\nu$ the proof-term obtained by replacing the redex $\mu$ by the proof-variable $\alpha$ in $\pi$. We have $\pi = [\mu/\alpha]\nu \ra [\rho/\alpha]\nu = \pi'$. The maximal length of a reductions sequence from $\pi'$ is less than or equal to $m-1$, then there exists, by induction hypothesis, $k \leq m-1$ such that $\pi' =  [\rho/\alpha]\nu \in Cl^{k}(A)_\phi$.  Thus $\pi \in Cl^{k+1}(A)_\phi$ with $k+1 \leq m$.
  \end{proof}

\medskip

\noindent In the following, we prove that $Cl(.)_.$ forms an interpretation that is adapted to the connective $\Ra$. We first prove a useful lemma (as explained in section \ref{otcr3p}).

\begin{lemma}
\label{lambdacl}\hfill\\
	For all propositions $A$, $B$ and environments $\phi$, 
	proof-terms $\pi$ and proof-variables $\alpha,\beta$ such that $\alpha$ does not occur free in $\pi$,
	\\if $\Delta \vdash \alpha : \phi A$ and $(\alpha/\beta)\pi \in Cl(B)_\phi$ 
	then $\lambda \beta.\pi \in Cl(A \Ra B)_\phi$.
\end{lemma}

\proof There exists a minimal $k$ such that \mbox{$(\alpha/\beta)\pi \in Cl^k(B)_\phi$.} By induction on $k$.
\begin{iteMize}{$\bullet$}
	\item if $k=0$, then $\Delta \vdash (\alpha/\beta) \pi : \phi B$, hence $\Delta \vdash  \lambda \beta.\pi : \phi(A \Ra B)$ and  we conclude that $\lambda \beta.\pi \in Cl^0(A \Ra B)_\phi$.
	
	\item if $k>0$, then $(\alpha/\beta)\pi = [\mu_i/\alpha_i]_{i}\nu$ with each $\mu_i \in \Omega$ and such that for all $(\rho_i)_{i }$ with $(\mu_i)_{i} \ra (\rho_i)_{i}$, we have
			$[\rho_i/\alpha_i]_{i}\nu \in Cl^{k-1}(B)_\phi$.  (H)
			\\ Let us write $\nu' = (\beta/\alpha)\nu$ and $\mu_i'=(\beta/\alpha)\mu_i$ for each $i$.  
			\\We have $\pi = (\beta/\alpha)(\alpha/\beta)\pi = (\beta/\alpha)([\mu_i/\alpha_i]_i\nu) = [(\beta/\alpha)\mu_i/\alpha_i]_i((\beta/\alpha)\nu) = [\mu_i'/\alpha_i]_i\nu'$ since $\alpha$ is not free in $\pi$. 
			\\Hence $\lambda \beta.\pi = \lambda \beta. ([\mu_i'/\alpha_i]_i\nu') = [\mu_i'/\alpha_i]_i(\lambda\beta.\nu')$ ($[./.]$ is substitution {\em with} capture). 
			\\ We can notice that the $\mu_i'$ belong to $\Omega$ since the $\mu_i$ do.
			\\Let $(\rho_i')_i$ such that $\mu_i' \ra \rho_i'$ for each $i$. 
			\\Notice that we have $(\alpha/\beta)\nu' = (\alpha/\beta)(\beta/\alpha)\nu = \nu$ and $(\alpha/\beta)\mu_i' = (\alpha/\beta)(\beta/\alpha)\mu_i = \mu_i$ for each $i$, since $\beta$ is not free in $(\alpha/\beta)\pi = [\mu_i/\alpha_i]_{i}\nu$. 
			\\Hence for each $i$, $\mu_i = (\alpha/\beta)\mu_i'  \ra (\alpha/\beta)\rho_i'$, thus \mbox{$[(\alpha/\beta)\rho_i'/\alpha_i]_{i}\nu \in Cl^{k-1}(B)_\phi$ by (H). }
			\\ Since $\nu = (\alpha/\beta)\nu'$, we have $(\alpha/\beta)([\rho_i'/\alpha_i]_{i}\nu') \in Cl^{k-1}(B)_\phi$ thus $\lambda \beta.([\rho_i'/\alpha_i]_{i}\nu')$ belongs to $Cl(A\Ra B)_\phi$ by induction hypothesis.
			\\ Finally, $[\rho_i'/\alpha_i]_i(\lambda\beta.\nu') \in Cl(A \Ra B)_\phi$ for all $(\rho_i)_i$ family of respective reducts of the $\mu_i$. Hence $\lambda\beta.\pi =  [\mu_i'/\alpha_i]_i(\lambda\beta.\nu') \in Cl(A \Ra B)_\phi$ since that set satisfies (CR$_3'$).\qed
	
%
%
%
%
\end{iteMize}\medskip

\begin{lemma} \hfill

\label{clramorph}
	For all propositions $A,B$ and environments $\phi$,  $Cl(A \Ra B)_\phi = Cl(A)_\phi \tra Cl(B)_\phi$.
\end{lemma}
	
\proof \hfill

\begin{iteMize}{$\subseteq$}
\item[$\subseteq$]  Let $\pi \in Cl(A \Ra B)_\phi$,  
				 \\then $\pi \in SN$ by (CR$_1$). Moreover there exists  (a minimal) $ k \in \N$, such that \mbox{$\pi \in Cl^k(A \Ra B)_\phi$}.
				 Let $\pi' \in Cl(A)_\phi$, then there exists (a minimal) $ j \in \N$, such that $\pi' \in Cl^j(A)_\phi$. Let us show that $\pi \pi' \in Cl(B)_\phi$ 
				   by induction on $k+j$.
				\begin{iteMize}{$-$}
		 			 \item If $k+j=0$, then $\Delta \vdash \pi : \phi(A \Ra B)$ and $\Delta \vdash \pi' : \phi A$ hence $\Delta \vdash \pi \pi' : \phi B$ and $\pi\pi' \in Cl^0(B)_\phi$.
				   	  \item If $k>0$,  then there exists $\nu_\pi$, and $(\mu_i)_{i \leq n} \subseteq \Omega$, such that  $\pi = [\mu_i/\alpha_i]_{i}~\nu_\pi$ and for all
							  $(\rho_i)_{i} $ with
							  $(\mu_i)_i \ra (\rho_i)_i$,  we have \mbox{$[\rho_i/\alpha_i]_{i}~\nu_\pi \in Cl^{k-1}(A \Ra B)_\phi$}. 
							Therefore \linebreak $[\rho_i/\alpha_i]_{i}~(\nu_\pi~ \pi') =  [\rho_i/\alpha_i]_{i}~\nu_\pi ~\pi' \in Cl(B)_\phi$ by induction hypothesis. 
						   	Hence $\pi \pi' \in Cl(B)_\phi$ since it satisfies (CR$_3'$).
					 \item If $j>0$,  then there exists $\nu_{\pi'}$, and $(\mu_i)_{i\leq n} \subseteq \Omega$, such that 
					 		 $\pi' = [\mu_i/\alpha_i]_{i}~\nu_{\pi'} $ and for all $(\rho_i)_{i}$, with
							  $(\mu_i)_i \ra (\rho_i)_i$, we have \mbox{$[\rho_i/\alpha_i]_{i}~\nu_{\pi'} \in Cl^{j-1}(A)_\phi$}. 
							  \\Therefore $[\rho_i/\alpha_i]_{i}~(\pi~ \nu_{\pi'}) =  [\rho_i/\alpha_i]_{i}~\pi ~\nu_{\pi'} \in Cl(B)_\phi$ by induction hypothesis. 
						   Hence $\pi \pi' \in Cl(B)_\phi$ since it satisfies (CR$_3'$).
				  \end{iteMize}
		 \medskip
		 
		\item[$\supseteq$]  Let $\pi \in Cl(A)_\phi \tra Cl(B)_\phi$,
				then $\pi \in SN$ and for all $\pi' \in Cl(A)_\phi$, $\pi \pi' \in Cl(B)_\phi$.
				\begin{iteMize}{$-$}
					\item If $\pi$ is a proof-abstraction $\lambda \beta.\pi'$, let $\alpha$ be a proof-variable, not free in $\pi'$, such that $\Delta \vdash \alpha : \phi A$, then 
						$(\lambda \beta.\pi') \alpha \in Cl(B)_\phi$ and so does $(\alpha/\beta)\pi'$, by (CR$_2$).  Therefore $\pi$ belongs to  $Cl(A\Ra B)_\phi$ by lemma 
						\ref{lambdacl}.
					\item If $\pi$ is neutral and normal, let $\alpha$ be a proof-variable such that $\Delta \vdash \alpha : \phi A$, then $\pi \alpha \in Cl(B)_\phi$. 
						Moreover $\pi$ is neutral and normal, therefore $\pi \alpha$ is normal, hence, by lemma \ref{mink}, $\pi \alpha \in Cl^0(B)_\phi$, {\em i.e.} 
						$\Delta \vdash \pi \alpha : \phi B$, and $\Delta \vdash \pi : \phi (A \Ra B)$, thus $\pi \in Cl^0(A\Ra B)_\phi \subset Cl(A\Ra B)_\phi$. 
					
					\item Otherwise, $\pi \in SN$, is neutral and not normal. All its neutral normal and not neutral (more than one-step-) reducts belong to 
						 $Cl(A)_\phi \tra Cl(B)_\phi$  by (CR$_2$) and therefore to $Cl(A \Ra B)_\phi$ by the previous points. 
						 By repeatedly using the (CR$_3'$) property, we conclude that $\pi$ also belongs to $Cl(A \Ra B)_\phi$.\qed\smallskip
				\end{iteMize}
	\end{iteMize}

\noindent Now that we have proved that $Cl(.)_.$ has its values in $\cprime$ (when the considered theory is strongly normalizing) and that it is adapted to the connective $\Ra$, the last properties needed to prove that $Cl(.)_.$ is a model are first the substitution property and second, the fact that it is also adapted to the connective $\fa$. This latter property becomes simpler to prove when considering theories expressed in minimal deduction modulo \`a la Curry since $\tfa$ is now classical set-inclusion.

\begin{lemma}[Substitution]
\label{clsubst}\hfill\\
For all propositions $A$ term-variables $x$, terms $t$ and environments $\phi$,  \linebreak $Cl((t/x)A)_\phi = Cl(A)_{\phi + \langle x,t \rangle}$.
\end{lemma}

\proof 
$\pi \in Cl^0((t/x)A)_\phi~$ if and only if $~\Delta \vdash \pi : \phi(t/x)A = (t/x)\phi A = (\phi + \langle x,t \rangle) A$ \linebreak if and only if $\pi \in Cl^0(A)_{\phi + \langle x,t \rangle}$
(notice that term-substitutions commute).
\\ Hence $Cl^0((t/x)A)_\phi = Cl^0(A)_{\phi + \langle x,t \rangle}$ and  $Cl((t/x)A)_\phi = Cl(A)_{\phi + \langle x,t \rangle}$.
\qed

\medskip

\begin{lemma}
\label{clfamorph}\hfill\\
For all propositions $A$ term-variables $x$ and environments $\phi$,  
\\ $Cl(\forall x.A)_\phi = \tfa ~\{ Cl(A)_{\phi + \langle x,t \rangle}$, for $t$ term$\}$.
\end{lemma}

\proof\hfill
\begin{iteMize}{$\bullet$}
		 \item[$\subseteq$] Let $\pi \in Cl(\forall x.A)_\phi$, then there exists (a minimal) $k \in \N$ such that \mbox{$\pi \in Cl^k(\forall x.A)_\phi$.}
		 	By induction on $k$. 
			\begin{iteMize}{$-$}
		 		\item[-] If $k=0$, $\Delta \vdash \pi : \phi(\forall x.A)$, hence for all terms $t$, $\Delta \vdash \pi : (t/x)\phi(A)$ and $\pi \in Cl^0(A)_{\phi + \langle x,t \rangle}$ by lemma \ref{clsubst}.
				  \item[-] If $k>0$, then $\pi = [\mu_i/\alpha_i]_{i} \nu$, with $(\mu_i)_i \subseteq \Omega$ and such that for all $(\rho_i)_{i}$ with 
				  	$ (\mu_i)_i \ra (\rho_i)_i$, we have $[\rho_i/\alpha_i]_{i} \nu \in Cl^{k-1}(\forall x.A)_\phi \subseteq  \tfa ~\{ Cl(A)_{\phi + \langle x,t \rangle}$, for $t$ term$\}$, by induction hypothesis. Since all the $Cl(A)_{\phi + \langle x,t \rangle}$ satisfy (CR$_3'$), $\pi$ also belongs to each of those 
					$Cl(A)_{\phi + \langle x,t \rangle}$.
			\end{iteMize}
			\smallskip
			\item[$\supseteq$] As seen in lemma \ref{mink}, if $\pi$ belongs to some $Cl(B)_\psi$ then there exists some $k$, less or equal than the maximal length of reductions from $\pi$, such that $\pi \in Cl^k(B)_\psi$. Hence, in our case, if $\pi \in \tfa ~\{ Cl(A)_{\phi + \langle x,t \rangle}$, for $t$ term$\}$, then  there exists $k$ such that for all terms $t$, $\pi \in Cl^k(A)_{\phi + \langle x,t \rangle}$. We reason by induction on $k$. If $k = 0$, we have, in particular $\pi \in Cl^0(A)_{\phi + \langle x,x \rangle}$, i.e. $\Delta \vdash \pi : \phi A$ hence $\Delta \vdash \pi : \phi (\forall x. A)$ (of course we can suppose that $x$ is not bound in $\phi$). If $k > 0$, then we conclude by the fact that $Cl(\forall x.A)_\phi$ satisfies (CR$_3'$), as usual.\qed
\end{iteMize}

\medskip

\noindent We finally get the completeness result:

\begin{theorem}[Completeness]
	\label{comp}\hfill\\
	If a theory is strongly normalizing in minimal deduction modulo \`a la Curry, then $Cl(.)_.$ is a  $\cprime$-valued model of that theory.
	
\end{theorem}

\begin{proof} By lemmas  \ref{clc'}, \ref{clramorph} and \ref{clfamorph}.
\end{proof}

\noindent Notice that the strong normalization hypothesis is only used once in the proof of completeness, when proving that the interpretation of a proposition and an environment only contains strongly normalizing proof-terms.

\bigskip

\subsection{Semantic characterization of strong normalization}\hfill\\

\noindent From Theorems \ref{sound} and \ref{comp}, we prove that the existence of a $\cprime$-valued model is a sound and complete semantic criterion for strong normalization of theories expressed in minimal deduction modulo.

\begin{theorem}
\label{compsound}
	\hfill\\
	A theory is strongly normalizing in minimal deduction modulo \`a la Curry if and only if there exists a   $\cprime$-valued model of that theory.
\end{theorem}

\noindent Remind that we obtain this result for both confusing and non-confusing theories unlike what we obtained in \cite{pstt} and \cite{Theseoim} for minimal deduction modulo \`a la Church.

\medskip

\noindent  Moreover,  for a theory, strong normalization in Curry style entails strong normalization in Church style \cite{CouMiq}, hence this {\em Curry}-criterion is also a sound criterion for strong normalization of {\em Church-style} proof-terms. 
The existence of a (Curry) $\cprime$-valued model entails strong normalization of both Curry-style and Church-style proof-terms of the considered theory.

\medskip

\noindent We also conjecture that strong normalization in Church style and in Curry style are equivalent, in which case, the criterion provided in this paper would also be complete for strong normalization in minimal deduction \`a la Church.

\bigskip
\bigskip

\section*{Perspectives}

\noindent In this paper, we provide a sound and complete semantic criterion for strong normalization of theories expressed in minimal deduction \`a la Curry. To obtain that semantic criterion, we use the extension (CR$_3'$) of the usual (CR$_3$) property of reducibility candidates, that was introduced in \cite{pstt,Theseoim} to provide a sound and complete criterion for strong normalization in minimal deduction \`a la Church. Considering proof-terms \`a la Curry, instead of \`a la Church, greatly simplifies both the definition of that criterion and the proof of its completeness, since it allows to use the classical notion of pre-Heyting algebras and to define the interpretation of the universal quantification as usual intersection.

\bigskip

\noindent  This sound and complete semantic criterion for strong normalization brings up a second link between proof-theory and model-theory. G\"odel-Tarski first link allowed to prove fundamental theorems about consistency of first order theories. As mentioned in the introduction, it allowed to prove the independence of the axiom of choice from Zermelo-Fraenkel set theory, using model-theory techniques, like Fraenkel-Mostowski permutation and Cohen's forcing, to entail from ZF consistency, both ZFC and ZF$\neg$C consistencies. This G\"odel-Tarski link also allowed to prove the independence of the continuum hypothesis from ZFC \cite{GodelcontCH,CohenCH}, via the same method.

\medskip

\noindent The ambitious goal I propose to pursue is to use the semantic criterion for strong normalization defined in this paper, and reproduce G\"odel-Cohen method for obtaining independence (for strong normalization) of the axiom of choice from ZF. This implies different steps: first extend Dowek-Miquel embedding of Zermelo set theory in deduction modulo in order to embed the whole Zermelo-Fraenkel set theory. Second, extend our semantic criterion for strong normalization to the whole deduction modulo (using also Curry-style proof-terms for existential quantification). And finally build, from a $\cprime$ model of ZF, two other models of ZF, one that is compatible with the axiom of choice and one that is not. We could, for that purpose, add to the language Hilbert's operator of choice $\varepsilon$. Then we may try to adapt Fraenkel-Mostowski's permutation method in order to build a $\cprime$-model of ZF in which all interpretations of a formula of the form $\exists x.A$ are equal to the interpretation of $(\varepsilon(A)/x)A$. And finally, we may try to adapt Cohen's forcing method in order to build another $\cprime$-model of ZF in which the previous property is not satisfied.  
\medskip

\noindent Finally, this paper can also have practical spin-offs concerning proof assistants and proof checkers.  
Dedukti \cite{dedukti} is a universal proof-checker that can check now proofs produced by the proof-assistants Coq \cite{Coq} and HOL-Light \cite{hollight}. Dedukti is based on the formalism of $\lambda \Pi$-calculus modulo \cite{PTS}, the extension of deduction modulo with dependent types. Dedukti's checking relies on computation and therefore on strong normalization of the theory that expresses the logical formalism of some proof assistant. 
The present paper gives techniques to prove relative normalizations of theories that represent proof assistants. For example we could prove that the union of theories representing Coq and HOL-Light is strongly normalizing, given the fact that one of them is strongly normalizing. That would allow to check simultaneously proofs coming from both proof assistants and ensure that way that the associated developments in each proof assistant can be somehow combined.




%

\bigskip
\bigskip

\newpage
\bibliographystyle{plain}
\bibliography{Bibliography}

\end{document}